\documentclass[conference]{IEEEtran}
\IEEEoverridecommandlockouts

\usepackage{times}
\usepackage{array,multirow}
\usepackage{xspace}
\usepackage{xcolor}
\usepackage{color, colortbl}
\usepackage{balance}
\usepackage{paralist}
\usepackage{amsfonts,amssymb}
\usepackage{amsthm,amsmath}
\usepackage{graphicx}
\usepackage{pifont}
\usepackage{array}
\usepackage{booktabs}
\usepackage{tikz}
\usepackage{bm}

\let\labelindent\relax
\usepackage{enumitem}

\usepackage{todonotes}
\usepackage{fancyhdr,graphicx,amsmath,amssymb}
\usepackage{algorithm}
\usepackage[noend]{algpseudocode}
\usepackage{cleveref}

\newcommand\sysname{FMPC\xspace}

\newcommand\fourier{Fourier\xspace}
\newcommand\parseval{Parseval\xspace}

\newcommand\user[1]{\textit{#1}\xspace}
\newcommand\player{user\xspace}
\newcommand\players{users\xspace}

\newcommand\Players{Users\xspace}
\newcommand\node[1]{$node_{#1}$\xspace}
\newcommand\Node[1]{$Node_{#1}$\xspace}

\newcommand{\eg}{\textit{e.g.,}\@\xspace}
\newcommand{\ie}{\textit{i.e.,}\@\xspace}

\def\first{({\it i})\xspace}
\def\second{({\it ii})\xspace}

\def\fourth{({\it iv})\xspace}

\newtheorem{theorem}{Theorem}

\def\BibTeX{{\rm B\kern-.05em{\sc i\kern-.025em b}\kern-.08em
    T\kern-.1667em\lower.7ex\hbox{E}\kern-.125emX}}

\begin{document}
\title{\sysname : Secure Multiparty Computation from \fourier Series and \parseval's Identity}

\author{\IEEEauthorblockN{Alberto Sonnino}
\IEEEauthorblockA{University College London (UCL)\\
alberto.sonnino@ucl.ac.uk}}

\maketitle

\begin{abstract}
\sysname is a novel multiparty computation protocol of arithmetic circuits based on secret-sharing, capable of computing multiplication of secrets with no online communication; it thus enjoys constant online communication latency in the size of the circuit. \sysname is based on the application of \fourier series to \parseval's identity, and introduces the first generalization of \parseval's identity for \fourier series applicable to an arbitrary number of inputs. \sysname operates in a setting where \players wish to compute a function over some secret inputs by submitting the computation to a set of nodes, but is only suitable for the evaluation of low-depth arithmetic circuits. \sysname relies on an offline phase consisting of traditional preprocessing as introduced by established protocols like SPDZ, and innovates on the online phase that mainly consists of each node locally evaluating specific functions. \sysname paves the way for a new kind of multiparty computation protocols 
capable of computing multiplication of secrets as an alternative to circuit garbling and the traditional algebra introduced by Donald Beaver in 1991.
\end{abstract}

\section{Introduction} \label{sec:introduction}
Multiparty computation protocols allow multiple \players to compute some function of their combined secret inputs without revealing any additional information about their inputs other than the output of the function. \sysname is a secret-sharing based protocol for arithmetic circuits~\cite{SPDZ}; it operates in a setting where \players wish to compute a function over some secrets by submitting the computation to a set of nodes, and is only suitable for circuits with a low number of multiplications. The \players first secret-share their inputs by breaking them into multiple shares, and provide each node with one each. The nodes then perform additions and multiplications on these shares by local computations, and finally output the result of the computation. \sysname focuses on the computation of multiplication of secrets, and assumes that additions can be performed using traditional algebra as described by SPDZ~\cite{SPDZ}.

As previous secret-sharing based protocols~\cite{SPDZ,SPDZ2,mascot}, \sysname divides execution into an \emph{offline phase} and an \emph{online phase}. The offline phase is performed ahead of time and does not involve any \players secret input; the output of the computation is then evaluated during the online phase. Traditional secret-sharing based protocols are efficient to compute additions of secrets, but computing multiplication is expensive~\cite{SPDZ}; these are based on the algebra introduced by Donald Beaver~\cite{beaver1991efficient} relying on the existence of some additional secret-shared values called \emph{triples}, that are generated during the offline phase. Each node then broadcasts their shares of secrets blinded with these triples value. This causes high communication complexity during the online phase, especially for computations requiring many multiplications; their latency increases with the number of multiplications to evaluate.

\sysname is a novel secret-sharing technique to compute multiplication of secrets without requiring nodes to communicate with each other at all during the online phase; \sysname thus enjoys constant (and low)  online communication latency in the size of the circuit. This is achieved through the application of \fourier series to \parseval's identity. On the downside, \sysname cannot compose operations and is therefore only suitable to evaluate circuits with a small number of multiplications (see \Cref{sec:limitations}). \sysname relies on established preprocessing techniques for the offline phase, and makes the following contributions to the online phase:
\begin{itemize}
    \item \Cref{sec:construction} presents the mathematical construction behind \sysname by taking the example of a two-\player computation.
    \item \Cref{sec:example} provides a concrete instantiation of \sysname and shows a practical protocol execution.
    \item \Cref{sec:extension} introduces the first generalization of \parseval's identity for \fourier series applicable to an arbitrary number of inputs, and uses it to extend the two-\player computation scheme presented in \Cref{sec:construction} to a scheme supporting an arbitrary number of \players. At the best of our knowledge, this is the first secret-sharing multiparty computation protocol scaling to an arbitrary number of inputs that enables multiplication of secrets with no online communication.
\end{itemize}
\sysname is a first of its kind attempt to analytically model MPC and aims to trigger further debates towards a working system.

\section{Threat Model and Goals} \label{sec:model}
The following actors participate in a \sysname computation:
\begin{itemize}
\item \textbf{\Players:} End-user devices submit a computation over some secret inputs to a set of nodes; they wish to publish the output of a computation without revealing their secret inputs to anybody. Without loss of generality, we assume that each \player hold one secret input.
\item \textbf{Nodes:} Infrastructure executing the computation submitted by the \players.
\end{itemize}
We model the offline phase as executed by a trusted authority responsible to generate some scheme parameters and communicate them to the \players; this offline phase can be distributed using traditional techniques introduced by SPDZ~\cite{SPDZ} (see \Cref{sec:no-trusty}). 
\sysname assumes passive adversaries who follow the protocol specification but try to learn more than allowed about the \players secret inputs\footnote{We leave the extension of \sysname to active adversaries as future work; potentially adapting the MAC-based approach introduced by SPDZ~\cite{SPDZ}.}. Nodes can collude with each other as long as there is at least one honest non-colluding node. 
Under the above threat model, \sysname achieves the following design goals:
\begin{itemize}
\item \textbf{Private Computation} - Parties only learn the output of the computation.
\item \textbf{Non-Interactivity} - Nodes do not communicate with each other during the online phase to perform computations.
\end{itemize}

\section{Background} \label{sec:background}
We recall the theory of \fourier series and \parseval's identity, and the expression of some useful convergent sums analytically; Appendix~\ref{sec:fields} shows how to compute them numerically using finite fields.

\subsection{Convolution of \fourier Series} \label{sec:fourier}
We recall the \fourier series of the convolution between two functions $f(x)$ and $g(x)$ periodic on ($-l,\ l$). Assuming that $f(x)$ and $g(x)$ $\in{\mathbb L}^2[-l,l]$ (\ie $f(x)$ and $g(x)$ are square-integrable in the interval [$-l,\ l$]), their respective \fourier series representations read:
\begin{align}\label{eq:conv1}
&  f(x)=\frac{a_0}{2}+\sum_{n=1}^{+\infty}a_n\cos\Bigl(\frac{n\pi x}{l}\Bigr)+\sum_{n=1}^{+\infty}b_n\sin\Bigl(\frac{n\pi x}{l}\Bigr)\\
&  g(x)=\frac{\alpha_0}{2}+\sum_{n=1}^{+\infty}\alpha_n\cos\Bigl(\frac{n\pi x}{l}\Bigr)+\sum_{n=1}^{+\infty}\beta_n\sin\Bigl(\frac{n\pi x}{l}\Bigr)\nonumber
\end{align}
where the Fourier coefficients $(a_0,\ a_n,\ b_n)$ and $(\alpha_0,\ \alpha_n,\ \beta_n)$ (for $n=1,2,\cdots)$ are given below:
\begin{align}\label{eq:conv2}
a_0&=\frac{1}{l}\int_{-l}^{l}f(x)dx \; ; \quad  a_n=\frac{1}{l}\int_{-l}^{l}f(x)\cos\Bigl(\frac{n\pi x}{l}\Bigr)dx \\ \nonumber
b_n&=\frac{1}{l}\int_{-l}^{l}f(x)\sin\Bigl(\frac{n\pi x}{l}\Bigr)dx\\ \nonumber
\alpha_0&=\frac{1}{l}\int_{-l}^{l}g(x)dx \; ; \quad \alpha_n=\frac{1}{l}\int_{-l}^{l}g(x)\cos\Bigl(\frac{n\pi x}{l}\Bigr)dx \\
\beta_n&=\frac{1}{l}\int_{-l}^{l}g(x)\sin\Bigl(\frac{n\pi x}{l}\Bigr)dx\nonumber
\end{align}
The convolution function between $f(x)$ and $g(x)$ is defined as
\begin{equation}\label{eq:conv3}
\Phi(x)=\Bigl(f\star g\Bigr)(x)\equiv\frac{1}{l}\int_{-l}^lf(y)g(x-y)dy
\end{equation}
By inserting \Cref{eq:conv1} into \Cref{eq:conv3}, and by taking into account the following identities
\begin{align}\label{eq:conv4}
&\frac{1}{l}\int_{-l}^l\sin\bigl(\frac{n'\pi}{l}x\bigr)\sin\bigl(\frac{n\pi}{l}x\bigr)dx=\delta_{nn'}\\
&\frac{1}{l}\int_{-l}^l\cos\bigl(\frac{n'\pi}{l}x\bigr)\cos\bigl(\frac{n\pi}{l}x\bigr)dx=\delta_{nn'}\nonumber\\
&\int_{-l}^l\sin\bigl(\frac{n'\pi}{l}x\bigr)\cos\bigl(\frac{n\pi}{l}x\bigr)dx=0\nonumber
\end{align}
where $\delta_{nn'}$ denotes Kronecker's delta, we obtain the \fourier series of the convolution between two functions:
\begin{align}\label{eq:conv5}
&\Phi(x)=\frac{c_0}{2}+\sum_{n=1}^\infty c_n\cos\bigl(\frac{n\pi}{l}x\bigr)+\sum_{n=1}^\infty d_n\sin\bigl(\frac{n\pi}{l}x\bigr)
\end{align}
where
\begin{align}
&c_0=\frac{a_0\alpha_0}{2}\quad ;\quad c_n=a_n\alpha_n-a_n\beta_n\quad ; \quad d_n=b_n\alpha_n+b_n\beta_n\nonumber
\end{align}
We also recall that the convolution operation satisfies commutativity and associativity; these properties are used in \Cref{sec:extension} to scale \sysname to an arbitrary number of inputs.

\subsection{\parseval's Identity}\label{sec:parseval}
Let's assume two functions $f(x)$ and $g(x)$ $\in{\mathbb L}^2[-l,l]$ as defined in \Cref{eq:conv1}; defining the four vectors ${\bf A}$, ${\bf B}$, $\boldsymbol\alpha$ and $\boldsymbol\beta$ (for $n=1,2,\cdots)$ as below,
\begin{align}\label{eq:par1}
{\bf A}&=\{a_0/{\sqrt 2},\ a_n\} \ \  ; \ \  {\bf B}=\{b_n\} \\ \nonumber
{\boldsymbol\alpha}&=\{\alpha_0/{\sqrt 2},\ \alpha_n\}\ \ ;\ \ {\boldsymbol\beta=\{\beta_n\}} 
\end{align}
\parseval's identity~\cite{gradshteyn2014table} holds for $f(x)$ and $g(x)$:
\begin{align}\label{eq:par2}
&{\bf A}\cdot{\boldsymbol\alpha}+{\bf B}\cdot{\boldsymbol\beta}=\frac{1}{l}\int_{-l}^l f(x)g(x)dx
\quad \text{or}\\
&\Big(\frac{a_0\alpha_0}{2}+\sum_{n=1}^\infty a_n\alpha_n\Big)+\Big(\sum_{n=1}^\infty b_n\beta_n\Big)=\frac{1}{l}\int_{-l}^l f(x)g(x)dx\nonumber
\end{align} 
\parseval's identity only applies to two functions; \Cref{sec:generalized-parseval} presents our generalization of \parseval's identity that applies to an arbitrary number of functions used to extend \sysname to an arbitrary number of inputs.

\subsection{Convergent Sums}\label{sec:convergence}
\sysname requires the computation of scalar products of vectors with infinite components. It is therefore crucial that the infinite series produced by these scalar products are convergent, and that the results of these series can be computed efficiently and exactly (\ie analytically). For example, in case of two \players, \sysname requires the evaluation of the following convergent sums (see \Cref{sec:example}):
\begin{align}\label{eq:sum1}
& \sum_{n=1}^\infty\frac{1}{(\gamma^2-n^2)(\delta^2-n^2)} \quad ; \quad \sum_{n=1}^\infty\frac{n^2}{(\gamma^2-n^2)(\delta^2-n^2)}
\end{align}
These expressions can be easily calculated from the following well-known identity~\cite{gradshteyn2014table}
\begin{equation}\label{eq:sum2}
\Omega(\gamma)=\sum_{n=1}^\infty\frac{1}{(\gamma^2-n^2)} = \frac{\pi}{2\gamma}\cot(\pi\gamma)-\frac{1}{2\gamma^2}
\end{equation}
as below:
\begin{align}\label{eq:sum3}
&\sum_{n=1}^\infty\frac{1}{(\gamma^2-n^2)(\delta^2-n^2)} = \frac{1}{\delta^2-\gamma^2}\Bigr(\Omega(\gamma)-\Omega(\delta)\Bigr)\nonumber\\
&\qquad=\frac{\pi}{2(\delta^2-\gamma^2)}\Bigl(\frac{\cot(\pi\gamma)}{\gamma}-\frac{\cot(\pi\delta)}{\delta}\Bigr)-\frac{1}{2\gamma^2\delta^2}\nonumber \\
&\sum_{n=1}^\infty\frac{n^2}{(\gamma^2-n^2)(\delta^2-n^2)}=\frac{1}{\delta^2-\gamma^2}\Bigr(\gamma^2\Omega(\gamma)-\delta^2\Omega(\delta)\Bigr)\nonumber\\
&\qquad=\frac{\pi}{2(\delta^2-\gamma^2)}\Bigr(\gamma\cot(\pi\gamma)-\delta\cot(\pi\delta)\Bigr)\nonumber
\end{align}
\Cref{sec:example} illustrates that a convenient choice of the mask functions allows evaluating the infinite series (\ie the scalar products) analytically.

\section{Two-\players \sysname Construction} \label{sec:construction}
We present the mathematical constructions behind \sysname by illustrating a two-\players computation protocol; \Cref{sec:example} provides a concrete instantiation of this construction.

\subsection{Mathematical Construction} \label{sec:simple-construction}
\Cref{fig:simple-fmpc} presents a two-\players \sysname computation. We consider two \players, \user{Alice} holding a secret input $a$ and \user{Bob} holding a secret input $b$, wishing to compute the product $ab$ without revealing their secret inputs. The protocol operates on the public parameters $l$ and $q$ (with $0 < q < 1$); and on the two parametric functions $\phi_{\tau'}, \psi_{\sigma'} \in{\mathbb L}^2[-l,l]$ whose parameters are generated by the trusted authority \user{Trusty}; we refer to those functions as \emph{mask functions}. The protocol is divided in two phases: an \emph{offline phase} consisting of pre-computations that can be performed ahead of time as it is independent on the secret inputs, and an \emph{online phase} producing the output $ab$.

\paragraph{Offline phase} We model the offline phase as executed by a trusted authority \user{Trusty} (\Cref{sec:no-trusty} shows how to distribute the offline phase). \user{Trusty} generates at random $\tau$ and $\sigma$, and computes the \textit{normalization coefficient} $\eta$ given by
\begin{equation}\label{ab1}
\eta^{-1}(\tau,\sigma,l)=\frac{1}{l}\int_{-l}^l\phi_\tau(x)\psi_\sigma(x)dx
\end{equation}
and computes the following \emph{normalized mask-functions}:
\begin{equation}\label{ab2}
\widetilde{\phi}_\lambda(x)=\eta^q\phi_\tau(x)\ \; \ \ \widetilde{\psi}_\lambda(x)=\eta^{1-q}\psi_\sigma(x)\nonumber
\end{equation} 
where $\mathbb\lambda$ indicates the set of parameters $\lambda=(\tau,\sigma,l,q)$~(\ding{202}). Contrarily to traditional secret-sharing protocols like SPDZ~\cite{SPDZ}, \sysname pushes the complexity at the edges by offloading the offline phase to the users.

\paragraph{Online phase} \user{Trusty} sends ${\tilde{\phi}}_{\lambda}$ to \user{Alice} and ${\tilde{\psi}}_{\lambda}$ to \user{Bob}, who respectively compute $f$ and $g$:
\begin{align}\label{ab3}
&f(x)=a\widetilde{\phi}_\lambda(x) \quad ; \quad  g(x)=b\widetilde{\psi}_\lambda(x) 
\end{align} 
\user{Alice} computes the  vectors $\bf A$ and $\bf B$ from $f(x)$, and \user{Bob} computes the  vectors ${\boldsymbol\alpha}$ and ${\boldsymbol\beta}$ from $g(x)$ as defined by \Cref{eq:conv2,eq:par1}~(\ding{203}). \user{Alice} sends $\bf A$ to \node{1} and $\bf B$ to \node{2}; and \user{Bob} sends $\boldsymbol\alpha$ to \node{1} and $\boldsymbol\beta$ to \node{2}. As a result, \node{1} gathers the constant and cosine component of the \parseval's identity, and \node{2} gathers the sine component of the \parseval's identity~(\ding{203}). \Node{1} outputs $(\boldsymbol{A\cdot\alpha})$, and \node{2} outputs $(\boldsymbol{B\cdot \beta})$~(\ding{205}); anyone can compute $(\boldsymbol{A\cdot\alpha})+(\boldsymbol{B\cdot\beta})=ab$ according to \Cref{eq:par2}. The intuition behind the scheme is to decompose the product $ab$ into two components that are eventually added together to compute the final result; this reduces the problem of multiplication of secret to an addition, which is enabled by \parseval's identity. \Cref{sec:example} presents an end-to-end example calculation, with practical choices of mask functions. 

\begin{figure}[t]
\centering
\includegraphics[width=\linewidth,keepaspectratio]{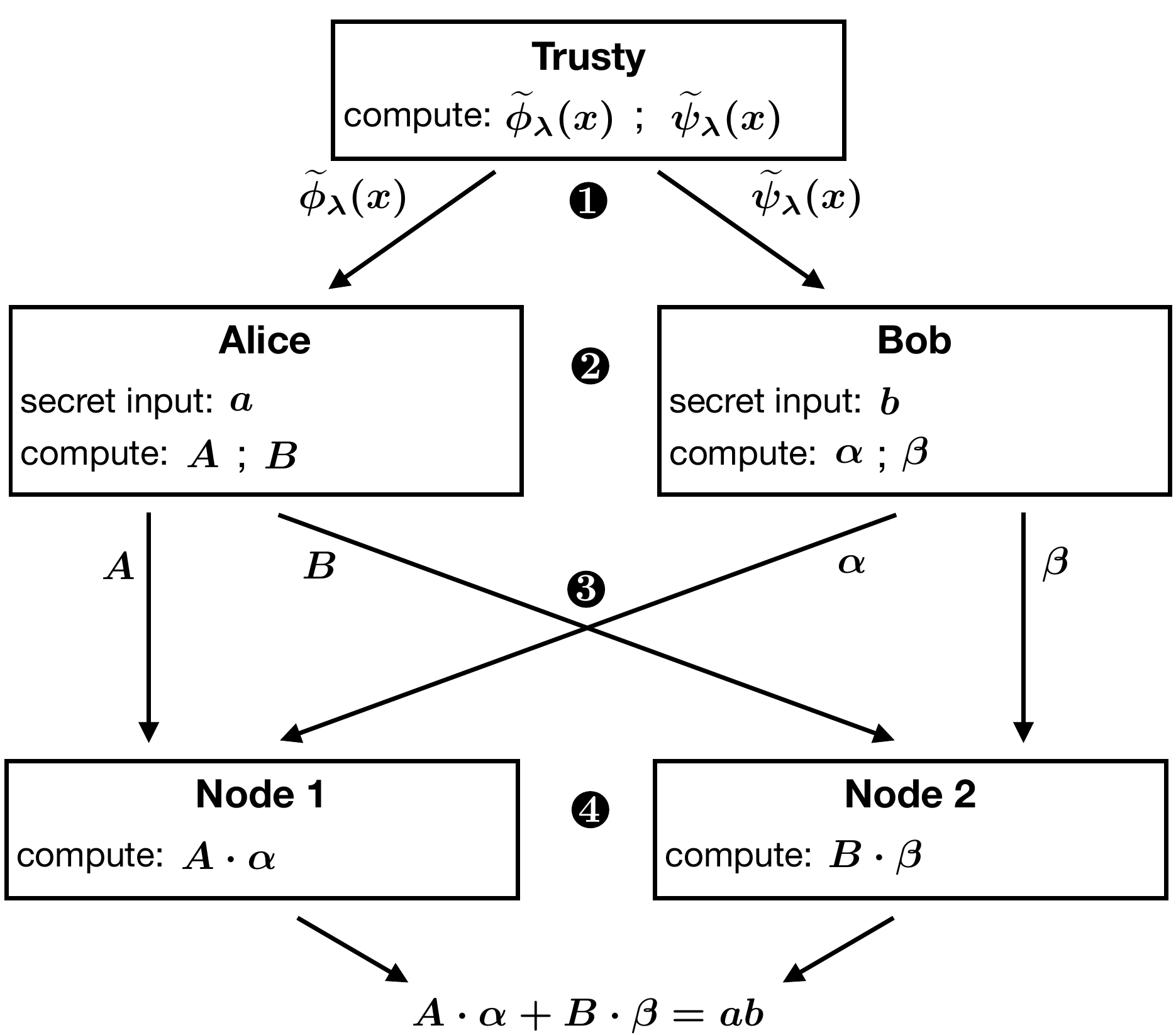}
\caption{\footnotesize Overview of \sysname execution. \user{Trusty} sends $\widetilde{\phi}_\lambda(x)$ to \user{Alice} and $\widetilde{\psi}_\lambda(x)$ to \user{Bob}~(\ding{202}). \user{Alice} computes $\bf A$ and $\bf B$, and \user{Bob} computes $\boldsymbol\alpha$ and $\boldsymbol\beta$ according to \Cref{eq:par1}~(\ding{203}); \user{Alice} sends $\bf A$ to \node{1} and $\bf B$ to \node{2}, and \user{Bob} sends $\boldsymbol\alpha$ to \node{1} and $\boldsymbol\beta$ to \node{2}~(\ding{204}). \Node{1} outputs $(\boldsymbol{A\cdot\alpha})$ and \node{2} outputs $(\boldsymbol{B\cdot\beta})$~(\ding{205}); anyone can compute $\boldsymbol{A\cdot\alpha}+\boldsymbol{B\cdot\beta}=ab$ according to \Cref{eq:par2}.}
\label{fig:simple-fmpc}
\end{figure}
%

\subsection{Decentralization of the Offline Phase} \label{sec:no-trusty}
We do not innovate on the offline phase, and rely on existing established solutions. The offline phase of \sysname randomly generates the parameters of the mask functions and computes the normalization coefficient.
\sysname may employ the same technique used by SPDZ~\cite{SPDZ} to generate multiplicative triples, which relies on somewhat homomorphic encryption; despite the simplicity of this approach, it incurs expensive public key cryptography and may lead to high cost. Mascot~\cite{mascot} overcomes this limitation by using oblivious transfer to generate the triples values during the offline phase. \Cref{sec:discussion} shows how to use the offline phase of those protocols to instantiate a practical \sysname computation.
Alternatively, \sysname  may rely on a semi-trusted authority to run the offline phase; the authority is then trusted to correctly generate those parameters and to not collude with the nodes, but never learns any information about the \players inputs.
\section{Instantiation of Two-\players \sysname Computation} \label{sec:example}
We illustrate a practical example of \sysname computation considering the following mask-functions:
\begin{align}\label{ex1}
&\phi_{\tau}(x) = \tau_1 \sin (\tau_2 x)+ \tau_3 \cos (\tau_4 x)\\
& \psi_{\sigma}(x) = \sigma_1 \sin (\sigma_2 x)+ \sigma_3 \cos (\sigma_4 x)\nonumber
\end{align} 
for parameters $\tau = (\tau_1, \tau_2, \tau_3, \tau_4)$ and $\sigma = (\sigma_1, \sigma_2, \sigma_3, \sigma_4)$. For simplicity, we set $ q=1/2$, and
\begin{align}\label{ex2}
&\tau_2=\sigma_4=\frac{\pi}{4l} \quad ; \quad 
\tau_4=\sigma_2=\frac{3\pi}{4l}
\end{align}
to obtain the following normalized mask-functions:
\begin{align}\label{ex3}
&{\widetilde\phi}_{\tau}(x) = \eta^{1/2}\Bigl(\tau_1 \sin \left(\frac{\pi}{4l} x\right)+ \tau_3 \cos\left(\frac{3\pi}{4l} x\right)\Bigr)\\
& {\widetilde\psi}_{\sigma}(x) = \eta^{1/2}\Bigl(\sigma_1 \sin\left(\frac{3\pi}{4l} x\right)+ \sigma_3 \cos\left(\frac{\pi}{4l} x\right)\Bigr)\nonumber\\
&{\rm with}\quad \eta=\frac{\pi}{2 (\tau_1\sigma_1+\tau_3\sigma_3)}\nonumber
\end{align} 

\subsection{Protocol Execution} \label{sec:protocol-execution}
We show how the protocol illustrated in \Cref{fig:simple-fmpc} executes using the mask functions given by \Cref{ex1}.

\paragraph{Offline phase} \Cref{alg:ex-offline} illustrates the offline phase; \user{Trusty} generates at random $(\tau_1, \tau_3, \sigma_1, \sigma_3)$; computes the normalization coefficients $\eta$; and sends $(\tau_1,\tau_3,\eta)$ to \user{Alice} and $(\sigma_1,\sigma_3,\eta)$ to \user{Bob}~(\ding{202}).

\paragraph{Online phase} \Cref{alg:ex-online} illustrates the online phase; \user{Alice} computes $(a_0,\widehat{a}_0)$, and \user{Bob} computes $(\alpha_0,\widehat{\alpha}_0)$~(\ding{203}). \user{Alice} sends $a_0$ to \node{1} and $\widehat{a}_0$ to \node{2}; \user{Bob} sends $\alpha_0$ to \node{1} and $\widehat{\alpha}_0$ to \node{2}~(\ding{204}). \Node{1} has all information it needs to compute $a_n$ and $\alpha_n$ (see \Cref{ex4} of \Cref{sec:analytic-computation})---those can be computed from the mere knowledge of $a_0$ and $\alpha_0$---and outputs $s_1=a_0\alpha_0/2+\sum_{n=1}^\infty a_n\alpha_n$. In practice, \node{1} only evaluates and outputs $s_1=(3\pi/16)a_0\alpha_0$ (see \Cref{ex6} of \Cref{sec:analytic-computation}). 
Similarly, \node{2} has all information to compute $b_n$ and $\beta_n$ (those can be computed from $\widehat{a}_0$ and $\widehat{\alpha}_0$), and outputs $s_2=\sum_{n=1}^\infty b_n\beta_n$; in practice, \node{2} simply outputs $s_2=(3\pi/16)\widehat{a}_0\ \widehat{\alpha}_0$ (see \Cref{ex7} of \Cref{sec:analytic-computation}). Anyone can compute $s_1+s_2=ab$, which follows from \Cref{eq:par2}~(\ding{205}).

All operations are performed over a finite field $\mathbb{F}^N_{p^m}$ where $p$ is prime, $m$ is integer and $N=p^m-1$; addition, multiplication, and the modular inverse are implemented by modular arithmetic $mod\ N$, that is $0\leq n\leq N-1$.

\begin{algorithm}[t]
\begin{algorithmic}[1]
\Procedure{Trusty}{}
\State $\tau_1, \tau_3, \sigma_1, \sigma_3 \gets random \text{ from } \mathbb{F}_{p^m}$
\State compute $\eta = \frac{\pi}{2(\tau_1\sigma_1+\tau_3\sigma_3)} \text{ mod } N$
\State send $(\tau_1,\tau_3,\eta)$ to \user{Alice}
\State send $(\sigma_1,\sigma_3,\eta)$ to \user{Bob}
\EndProcedure
\end{algorithmic}
\caption{\sysname example computation -- Offline phase}
\label{alg:ex-offline}
\end{algorithm}

\begin{algorithm}[t]
\begin{algorithmic}[1]
\Procedure{Alice}{$\tau_1,\tau_3,\eta$}
\State compute $a_{0} = \frac{4\sqrt{2}}{3\pi} a \tau_3 \eta^{1/2}\text{ mod } N$
\State compute $\widehat{a_{0}} = \frac{4\sqrt{2}}{3\pi} a \tau_1 \eta^{1/2}\text{ mod } N$
\State send $a_{0}$ to \node{1} and $\widehat{a_{0}}$ to \node{2}
\EndProcedure
\Procedure{Bob}{$\sigma_1,\sigma_3,\eta$}
\State compute $\alpha_{0} = \frac{4\sqrt{2}}{\pi} b \sigma_3 \eta^{1/2}\text{ mod } N$
\State compute $\widehat{\alpha}_0 = \frac{4\sqrt{2}}{\pi} b \sigma_1 \eta^{1/2}\text{ mod } N$
\State send $\alpha_{0}$ to \node{1} and $\widehat{\alpha}_0$ to \node{2}
\EndProcedure
\Procedure{Node1}{$a_{0}, \alpha_{0}$}
\State output $s_1 = (3\pi/16)a_0\alpha_0\text{ mod } N$
\EndProcedure
\Procedure{Node2}{$\widehat{a_{0}}, \widehat{\alpha}_0$}
\State output $s_2 = (3\pi/16)\widehat{a}_0\widehat{\alpha}_0\text{ mod } N$
\EndProcedure
\end{algorithmic}
\caption{\sysname example computation -- Online phase}
\label{alg:ex-online}
\end{algorithm}

%

\subsection{Correctness of the Computation} \label{sec:analytic-computation}
We compute the normalization coefficients and the normalized mask functions according to \Cref{ab1} and (\ref{ab2}); and the functions $f(x)$ and $g(x)$ according to \Cref{ab3}. All computations are performed using Wolfram Mathematica\footnote{http://www.wolfram.com/mathematica/} 11.2, we release our script as open source\footnote{
https://gist.github.com/asonnino/7d3abd570736d13bddf61fa429692983
}. The \fourier coefficients $(a_0,\ a_n,\ b_n)$ and $(\alpha_0,\ \alpha_n,\ \beta_n)$ are then given below (for $n=1,2,\cdots)$:
\begin{align}\label{ex4}
&a_0 = \frac{4\sqrt{2}}{3\pi}a\tau_3\eta^{1/2} \quad ; \quad
a_n = \frac{9(-1)^n}{9-16n^2}a_0\\ \nonumber
&b_n = \frac{12(-1)^n n}{1-16n^2}\widehat{a}_0\\ \nonumber
&\alpha_0 = \frac{4\sqrt{2}}{\pi}b\sigma_3\eta^{1/2} \quad ; \quad
\alpha_n = \frac{(-1)^n}{1-16n^2}\alpha_0 \\ \nonumber
&\beta_n = \frac{4(-1)^n n}{9-16n^2}\widehat{\alpha}_0 \nonumber
\end{align} 
where:
\begin{align}\label{ex5}
&\widehat{a}_0=\frac{4\sqrt{2}}{3\pi}a\tau_1\eta^{1/2} \quad ; \quad \widehat{\alpha}_0=\frac{4\sqrt{2}}{\pi}b\sigma_1\eta^{1/2} \\\nonumber
&\eta=\frac{\pi}{2(\tau_1\sigma_1+\tau_3\sigma_3)}
\end{align}
We can easily check \parseval's identity; \node{1} computes
\begin{align}\label{ex6}
& \frac{a_0\alpha_0}{2}+\sum_{n=1}^\infty a_n\alpha_n = \\ \nonumber
& a_0\alpha_0\Big(\frac{1}{2}+\frac{9}{256}\sum_{n=1}^\infty\frac{1}{[(3/4)^2-n^2][(1/4)^2-n^2]}\Big)=\frac{3\pi}{16}a_0\alpha_0
\end{align} 
and \node{2} computes
\begin{align}\label{ex7}
& \sum_{n=1}^\infty b_n\beta_n = \\ \nonumber & \frac{3}{16}\widehat{a}_0\widehat{\alpha}_0\sum_{n=1}^\infty\frac{n^2}{[(3/4)^2-n^2][(1/4)^2-n^2]}= \frac{3\pi}{16}\widehat{a}_0\widehat{\alpha}_0
\end{align} 
\Cref{ex6} and \Cref{ex7} are computed by evaluating the convergent sums given by \Cref{eq:sum1} of \Cref{sec:background}. By adding \Cref{ex6} to \Cref{ex7}, we finally get:
\begin{equation}\label{ex8}
\frac{a_0\alpha_0}{2}+\sum_{n=1}^\infty a_n\alpha_n+\sum_{n=1}^\infty b_n\beta_n= \frac{3\pi}{16}(a_0\alpha_0+\widehat{a}_0\widehat{\alpha}_0)=ab
\end{equation}

\subsection{Security Analysis} \label{sec:security-analysis}
We show that no  adversary can retrieve the secret inputs $a$ and $b$ from the knowledge of $(a_0, \alpha_0, \widehat{a}_0, \widehat{\alpha}_0, \eta)$. We assume passive adversaries; \ie they follow the protocol specification but try to learn more than allowed (see \Cref{sec:model}). Informally, the adversary possesses five equations, \ie the expressions of $(a_0, \alpha_0, \widehat{a}_0, \widehat{\alpha}_0, \eta)$, and six unknown, \ie $(a,b,\tau_1,\tau_3,\sigma_1,\sigma_3)$. The adversary thus holds fewer equations than unknowns, which make it information-theoretically impossible to recover any unknown value. \Cref{th:node-secrecy} presents this result more formally.

\begin{theorem}\label{th:node-secrecy}
The scheme presented in \Cref{sec:protocol-execution} achieves perfect secrecy against a passive adversary holding $\theta=(a_0, \alpha_0, \widehat{a}_0, \widehat{\alpha}_0, \eta)\in\mathbb{F}_{p^m}$; \ie for all distribution of $a,b\in\mathbb{F}_{p^m}$ and for all $\theta=(a_0, \alpha_0, \widehat{a}_0, \widehat{\alpha}_0, \eta)\in\mathbb{F}_{p^m}$, we have $Pr[a|\theta]=Pr[a]$ and  $Pr[b|\theta]=Pr[b]$.
\end{theorem}
\begin{proof}
Let us first consider the input $a$. For any $a,\theta\in\mathbb{F}_{p^m}$ we introduce the conditional probability $Pr[a|\theta]$ in terms of the joint probability
\begin{align} \label{eq:th1}
& Pr[a|\theta] = \frac{Pr[a\wedge\theta]}{Pr[\theta]} = \frac{Pr[\theta|a]Pr[a]}{Pr[\theta]}
\end{align}
where the second equality is guaranteed by Bayes' theorem. It is useful to compute $Pr[\theta]$ using the law of total probability. Conditioning over all $a\in\mathbb{F}_{p^m}$ gives
\begin{align}\label{eq:th2}
& Pr[\theta] = \sum_{a\in\mathbb{F}_{p^m}} Pr[\theta|a]Pr[a]
\end{align}
Removing all constant and known factors from the expression of $\theta$, we get
\begin{align}\label{eq:th3}
Pr[\theta|a] &= Pr[\widehat{\theta}]
\end{align}
where,
\begin{align}\label{eq:th3b}
&\widehat{\theta}=(\tau_1\eta^{1/2}\wedge\tau_3\eta^{1/2}\wedge b\sigma_1\eta^{1/2}\wedge b\sigma_3\eta^{1/2}\wedge\eta)
\end{align}
which is independent of $a$. Hence,
\begin{align}\label{eq:th4}
& Pr[\theta] = Pr[\theta|a] \sum_{a\in\mathbb{F}_{p^m}} Pr[a] = Pr[\theta|a]=Pr[\widehat{\theta}] 
\end{align}
Plugging \Cref{eq:th3} and \Cref{eq:th4} into \Cref{eq:th1}, we get $Pr[a|\theta]=Pr[a]$. The same reasoning applies to the input $b$.
\end{proof}

This implies that nodes are not able to recover the \players inputs even if they collude (but multiple nodes are still required to handle additions of secrets, as in SPDZ~\cite{SPDZ}). 


\subsection{Discussion} \label{sec:discussion}
We discuss convenient choice of mask functions, distribution of the offline phase, and extension to multiple nodes.

\paragraph{Convenient choice of mask functions} Even though \sysname applies to any kind of square-integrable functions, a convenient choice of family of mask functions (in the case of two players) is \{$\xi_1^{(j)} \sin (\xi_2^{(j)} x) + \xi_3^{(j)} \cos(\xi_4^{(j)} x)$\} where parameters $\xi_i^{(j)}$ are randomly chosen (with $i = 1,\ 2,\ 3,\ 4$ and \players $j=1,2$). The parameters $\xi_2^{(j)}$ and $\xi_4^{(j)}$ (with $j=1,2$) are public, and it is convenient to set them to  $\xi_2^{(1)}=\xi_4^{(2)}=\pi/(4l)$ and $\xi_4^{(1)}=\xi_2^{(2)}=(3\pi)/(4l)$ (see \Cref{ex3}). The main advantage of this family of mask-functions is that they forgo the need to resort to numerical calculations to compute the contributions of \parseval's identity---calculating the numerical sums of the \parseval's identity is never needed---\players simply evaluate them using the analytic expressions provided in \Cref{sec:convergence}. We can easily observe that it is possible to select mask-functions allowing to perform all calculations analytically even for a large number of \players; mask functions composed of sums of sine and cosine ensures convergence, and can be evaluated using  expressions similar to those given in \Cref{sec:convergence}.

\paragraph{Distribution of the offline phase} Established protocols like SPDZ require the generation of multiplicative triplets during the offline phase; \ie they provide a functionality to generate three elements $(x,y,z)$ such that $xy=z$ in a distributed manner. \sysname may execute twice this functionality to generate $(\tau_1,\sigma_1)$ such that $\tau_1\sigma_1=\eta_1$, and $(\tau_3,\sigma_3)$ such that $\tau_3\sigma_3=\eta_3$; and then simply compute:
\begin{equation}
\eta=\frac{\pi}{2(\eta_1+\eta_3)}=\frac{\pi}{2(\tau_1\sigma_1+\tau_3\sigma_3)}
\end{equation}

\paragraph{Extension to an arbitrary number of nodes} \Cref{sec:security-analysis} shows that colluding nodes cannot retrieve the \players inputs; multiple nodes are only required to handle addition of secrets. However, we can easily extend \sysname to an arbitrary number of nodes $m$ as we can always split the calculations of \parseval's identity into an arbitrary number parts. This can be accomplished in many ways; for instance we may split the vectors {\bf A}, {\bf B}, $\boldsymbol{\alpha}$, and $\boldsymbol{\beta}$ in several contributions, by requiring each of the $m$ nodes to perform only a specific part of the scalar products, under the constraint that the sum of their output matches with the final values of the scalar products of \parseval's identity. Note that for the example depicted in \Cref{ex4} (\Cref{sec:analytic-computation}), we can simply split the scalars $a_0$ and $\alpha_0$ into $m$ shares in such a way that the sum of the contributions coincides with the final scalar products (by applying appropriate normalization).

\section{Extension to Multiple Players} \label{sec:extension}
We introduces the first generalization of \parseval's identity for \fourier series applicable to an arbitrary number of inputs, and uses it to extend the two-\player computation scheme presented in \Cref{sec:construction} to an arbitrary number of \players.

\subsection{Generalization of \parseval's Identity} \label{sec:generalized-parseval}
We present the generalization of \parseval's identity for \fourier series applicable to $n$ inputs. \parseval's identity traditionally applies only to two functions; we overcome this drawback by using the convolution operation between two functions. We illustrate \parseval's identity for three inputs, which can easily be generalized for an arbitrary number of inputs. \Cref{sec:complete-construction} leverages these considerations to build the $n$-\players \sysname protocol.

Firstly we observe that in the case of two \players, \parseval's identity may be cast into the following form
\begin{align}\label{eq:gpar2-3ainputs}
&\frac{1}{2}a_0\alpha_0+\frac{1}{2}\boldsymbol{(A+B)\cdot(\alpha+\beta)}+\frac{1}{2}\boldsymbol{(A-B)\cdot(\alpha-\beta)}=\nonumber \\
&\frac{1}{l}\int_{-l}^l f(x)g(x)dx
\end{align}
Let's now consider three inputs,  $f(x)$, $g(x)$ and $h(x)$ with \fourier series representations given by \Cref{eq:conv1} and by
\begin{equation}\label{eq:gpar1}
h(x)=\frac{\gamma_0}{2}+\sum_{n=1}^\infty \gamma_n\cos\bigl(\frac{n\pi}{l}x\bigr)+\sum_{n=1}^\infty \varrho_n\sin\bigl(\frac{n\pi}{l}x\bigr)\\ 
\end{equation}
respectively; the generalized \parseval's identity reads:
\begin{align}\label{eq:gpar2-3inputs}
&\frac{1}{2}a_0\alpha_0\gamma_0+ \frac{1}{2}\boldsymbol{(A+B)\cdot(\alpha+\beta)\cdot(\gamma+\varrho)} +\\ \nonumber
&\frac{1}{2}\boldsymbol{(A-B)\cdot(\alpha-\beta)\cdot(\gamma-\varrho)}= \\ \nonumber
& \frac{1}{l}\int_{-l}^lf(x)(g\star h)(x)dx+\frac{1}{l}\int_{-l}^lg(x)(f\star h)(x)dx+
\\ \nonumber 
&\frac{1}{l}\int_{-l}^lh(x)(f\star g)(x)dx-\frac{2}{l}\int_{-l}^l\hat{h_c}(x)(\hat{f_c}\star \hat{g_c})(x)dx
\end{align}
or
\begin{align}\label{eq:gpar3}
&\frac{a_0\alpha_0 \gamma_0}{2}+\frac{1}{2}\sum_{n=1}^\infty (a_n+b_n)(\alpha_n+\beta_n)(\gamma_n+\varrho_n)+
\\ \nonumber
&\frac{1}{2}\sum_{n=1}^\infty (a_n-b_n)(\alpha_n-\beta_n)(\gamma_n-\varrho_n)= \\ \nonumber
& \frac{1}{l}\int_{-l}^lf(x)(g\star h)(x)dx+\\ \nonumber
&\bigg(\frac{1}{l}\int_{-l}^lg(x)(f\star h)(x)dx-\frac{1}{l}\int_{-l}^l\hat{g_c}(x)(\hat{f_c}\star \hat{h_c})(x)dx\bigg) +\\ \nonumber
\\ \nonumber 
&\bigg(\frac{1}{l}\int_{-l}^lh(x)(f\star g)(x)dx-\frac{1}{l}\int_{-l}^l\hat{h_c}(x)(\hat{f_c}\star \hat{g_c})(x)dx\bigg)
\end{align}\label{eq:gpar4}
Vectors $\mathbf A$, $\mathbf B$, $\boldsymbol{\alpha}$, $\boldsymbol{\beta}$ and $\boldsymbol{\gamma}$ and $\boldsymbol{\varrho}$ are respectively defined as
\begin{align}\label{eq:gpar5}
{\bf A}&=\{a_n\} \ \  ; \ \  {\bf B}=\{b_n\} \quad (n=1,2,\cdots) \\ \nonumber
{\boldsymbol\alpha}&=\{\alpha_n\}\ \ ;\ \ {\boldsymbol\beta=\{\beta_n\}}\quad(n=1,2,\cdots)\\ \nonumber
{\boldsymbol\gamma}&=\{\gamma_n\}\ \ ;\ \ {\boldsymbol\varrho=\{\varrho_n\}}\quad (n=1,2,\cdots)
\end{align}
and 
\begin{align}\label{eq:gpar6}
\hat{f}_c(x)&\equiv\frac{1}{2}\Bigl(f(x)+f(-x)\Bigr) \ \  ; \  \hat{g}_c(x)\equiv\frac{1}{2}\Bigl(g(x)+g(-x)\Bigr) \nonumber\\
\hat{h}_c(x)&\equiv\frac{1}{2}\Bigl(h(x)+h(-x)\Bigr)
\end{align}
We simply include $\gamma_0$, $\boldsymbol{(\gamma+\varrho)}$, and $\boldsymbol{(\gamma-\varrho)}$ in the left-hand side of the equation, and adapt the right-end side to match calculations. A mathematical formula for an arbitrary number of inputs can easily be obtained following the same logic.

\subsection{Mathematical Construction} \label{sec:complete-construction}
We extend the two-\player \sysname scheme presented in \Cref{sec:construction} to a scheme supporting $n$-\players.

\begin{figure}[t]
\centering
\includegraphics[width=\linewidth,keepaspectratio]{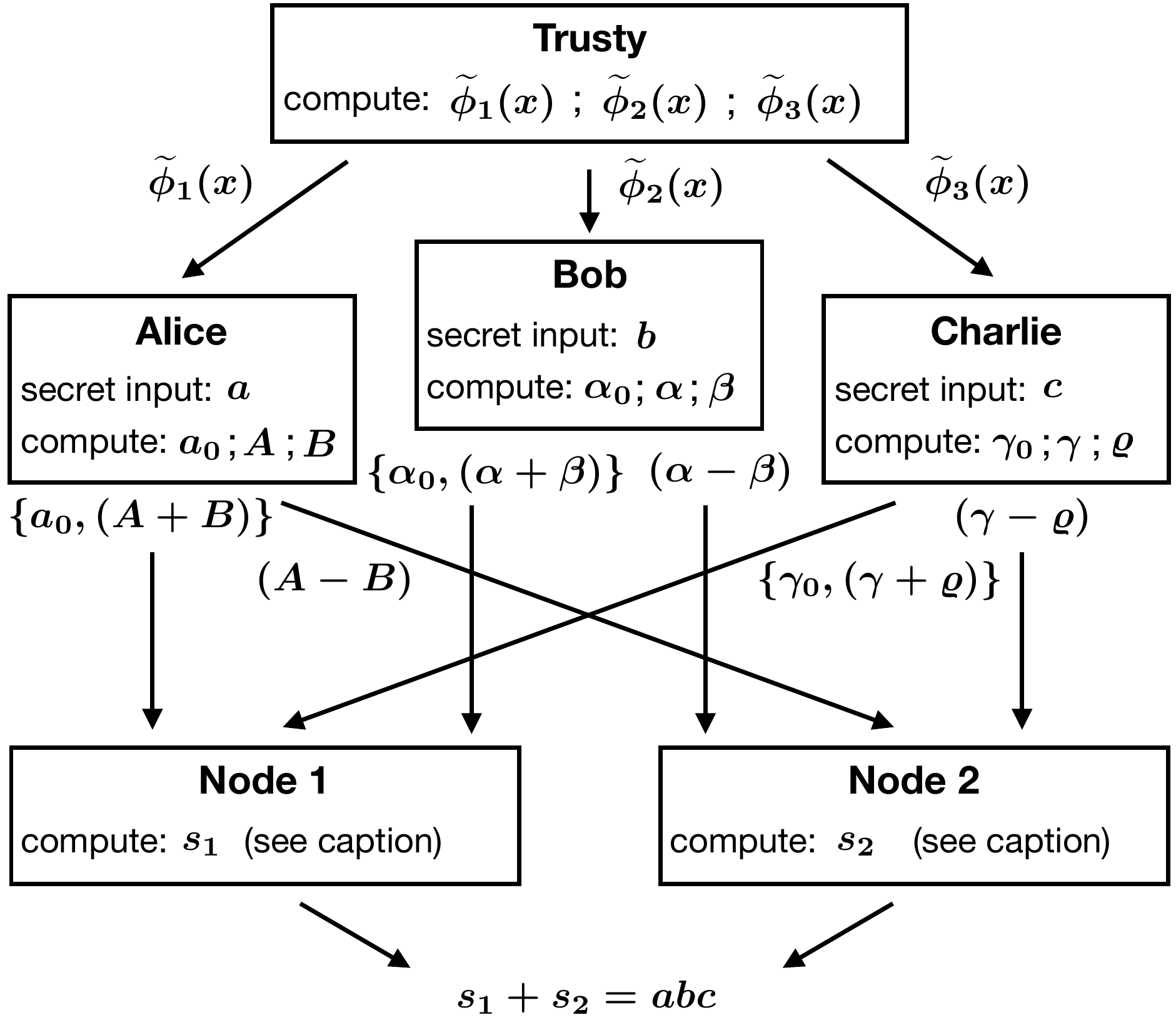}
\caption{\footnotesize Example execution of \sysname for 3 \players and 2 nodes. Each \player receives a normalized mask function from \user{Trusty}. \user{Alice} computes and sends $\{a_0,\bf (A+B)\}$ to \node{1} and $\bf (A-B)$ to \node{2}; \user{Bob} sends $\{\alpha_0,(\boldsymbol{\alpha+\beta})\}$ to \node{1} and $(\boldsymbol{\alpha-\beta})$ to \node{2}; and \user{Charlie} sends $\{\gamma_0,\boldsymbol{\gamma}+\boldsymbol{\varrho}\}$ to \node{1} and $(\boldsymbol{\gamma}-\boldsymbol{\varrho})$ to \node{2}. \Node{1} outputs $s_1 =1/2(a_0\alpha_0\gamma_0)+1/2 \boldsymbol{(A+B)\cdot(\alpha+\beta)\cdot(\gamma+\varrho)}$, and \node{2} outputs $s_2=1/2\boldsymbol{(A-B)\cdot(\alpha-\beta)\cdot(\gamma-\varrho)}$; according to \Cref{eq:gpar2-3inputs}, anyone can compute $s_1+s_2=abc$.}
\label{fig:fmpc}
\end{figure}

\paragraph{Offline phase} \user{Trusty} generates at random the parameters of $n$ mask-functions $(\phi_1, \cdots, \phi_n)$; it then computes the $n$ normalization coefficients similarly to \Cref{ab1}, and uses them to compute the normalized mask-functions $(\widetilde{\phi_1}, \cdots, \widetilde{\phi_n})$ as shown in \Cref{ab2}. This is analogue to the offline phase of the protocol presented in \Cref{sec:simple-construction}, except that we now consider $n$ mask-functions instead of two.

\paragraph{Online phase} \user{Trusty} sends a normalized mask function to each \player; they compute $(f_1, \cdots, f_n)$ using their secret inputs $(a_1,\cdots, a_n)$, and their \fourier coefficients according to \Cref{ab3} and \Cref{eq:conv2}. Similarly to \Cref{sec:simple-construction}, \players send the constant and cosine component of \parseval's identity to \node{1}, and the sine component to \node{2}; therefore the protocol can always be executed with two nodes. Each node then computes and outputs the scalar product of the \players coefficients vector, and the product $\prod_{i=1}^n a_i$ is computed by summing the output of each node according to the generalized \parseval's identity presented in \Cref{sec:generalized-parseval}.

\Cref{fig:fmpc} shows an example of execution of \sysname for three \players. Each \player, \user{Alice}, \user{Bob} and \user{Charlie} receives a normalized mask-function from \user{Trusty}. In this case the normalization coefficient is given by 
\begin{align}\label{eq:abc1}
&\eta^{-1}=\frac{1}{l}\int_{-l}^l\bigg(\phi_1(x)(\phi_2\star\phi_3)(x)+\phi_2(x)(\phi_1\star\phi_3)(x)+\nonumber\\
&\phi_3(x)(\phi_1\star\phi_2)(x)-2\hat{\phi}_{3c}(x)(\hat{\phi}_{1c}\star\hat{\phi}_{2c})(x)\bigg)dx
\end{align}\label{eq:abc2}
with
\begin{equation}\label{eq:abc3}
\hat{\phi}_{ic}\equiv\frac{1}{2}\Bigl(\phi_i(x)+\phi_i(-x)\Bigr)\quad{\rm with}\quad (i=1,2,3)
\end{equation}
and the three normalized mask-functions read
\begin{equation} \nonumber
\widetilde{\phi}_1(x)=\eta^{q_1}\phi_1(x)\ \; \widetilde{\phi}_2(x)=\eta^{q_2}\phi_2(x)\ \; \widetilde{\phi}_3(x)=\eta^{1-q_1-q_2}\phi_3(x)
\end{equation}
where $q_1$ and $q_2$ are two positive real numbers subject to the condition $0<q_1+q_2<1$. \user{Alice} locally computes $\{a_0,\bf A\}$ and $\bf B$; \user{Bob} computes $\{\alpha_0,\boldsymbol{\alpha}\}$ and $\boldsymbol{\beta}$; and \user{Charlie} computes $\{\gamma_0,\boldsymbol\gamma\}$ and $\boldsymbol{\varrho}$ similarly to \Cref{eq:conv2} and \Cref{eq:par1}. \user{Alice} sends $\{a_0,\bf{A+B}\}$ to \node{1} and $\bf(A-B)$ to \node{2}; \user{Bob} sends $\{\alpha_0,(\boldsymbol{\alpha+\beta})\}$ to \node{1} and $(\boldsymbol{\alpha-\beta})$ to \node{2}; and \user{Charlie} sends $\{\gamma_0,(\boldsymbol{\gamma}+\boldsymbol{\varrho})\}$ to \node{1} and $(\boldsymbol{\gamma}-\boldsymbol{\varrho})$ to \node{2}. Finally, \node{1} outputs $s_1=\frac{1}{2}a_0\alpha_0\gamma_0+\frac{1}{2}\boldsymbol{(A+B)\cdot(\alpha+\beta)\cdot(\gamma+\varrho)}$, and \node{2} outputs $s_2=\frac{1}{2}\boldsymbol{(A-B)\cdot(\alpha-\beta)\cdot(\gamma-\varrho})$; following \Cref{eq:gpar2-3inputs}, anyone can compute $s_1+s_2=abc$.

\section{Related Works} \label{sec:related}
There are two main constructions of multiparty protocols: circuit garbling and secret-sharing. Circuit garbling involves encrypting keys in a specific order to simulate a circuit evaluation~\cite{applebaum2014garble}; secret-sharing based protocol as \sysname break the inputs among all nodes who use their shares to evaluate some function through local computations~\cite{bendlin2011semi,damgaard2013constant,nielsen2012new,lindell2015efficient}.

SPDZ~\cite{SPDZ} is one of the most notorious secret-sharing based multiparty computation protocol scaling to an arbitrary number of \players; SPDZ is secure against active adversaries using MACs to verify the integrity of computations, and does not require any kind of trusted third parties; it requires however expensive somewhat homomorphic encryption (SHE) to generate the triples used to compute multiplication of secrets. SPDZ2~\cite{SPDZ2} offers various improvements of the offline phase of SPDZ, and allows the MACs to be checked without revealing its key, thus allowing the MAC to be re-used after it is checked. Mascot~\cite{mascot} uses oblivious transfer rather than SHE to further improve performances of the offline phase and generate triples.

The literature following SPDZ mainly improves the offline phase, while \sysname innovates on the online phase.
Most multiparty protocols for arithmetic circuits based on secret-sharing that scale to an arbitrary number of \players are based on the algebra introduced by Donald Beaver~\cite{beaver1991efficient}. They thus require triples to compute multiplication of secrets and impose communication between nodes during the online phase; their online latency therefore increases with the number of multiplications to evaluate. \sysname comes with a different trade-off: \sysname nodes do not communicate during the online phase and thus enjoy constant (and low) online latency in the size of the circuit, at the cost of not supporting composition of operations (see \Cref{sec:limitations}) which makes \sysname only suitable to evaluate low-depth circuits. Established secret-sharing protocols face a trade-off between security and online latency---adding nodes improves security but increases latency. \sysname forgoes this trade-off since multiplications can always be performed by two nodes (see \Cref{sec:extension}); however its security rely on the choice of the mask functions.

\section{Limitations and Future Work} \label{sec:limitations}
\sysname has several limitations that are beyond  the  scope  of this work, and deferred to future work. 
\sysname \first does not support composition of operations. That is, while most established scheme~\cite{SPDZ,SPDZ2,mascot} can evaluate expressions like $(a+b)(c+d)$ with two additions and one multiplication, \sysname needs to distribute the operation and evaluate $(ac+ad+bc+bd)$. This limitation is problematic for large computations and makes \sysname suitable only to evaluate circuits with a relatively small number of multiplications. Other limitations are \second that the security and efficiency of the scheme rely on the choice of the mask functions. We also defer as future work \fourth adapting our scheme to withstand active adversaries, potentially adapting the MAC-based approach introduced by SPDZ~\cite{SPDZ}.

\section{Conclusions} \label{sec:conclusion}
\sysname is a novel secret-sharing multiparty computation protocol of arithmetic circuits that requires no online communication between nodes to compute multiplication of secrets; \sysname innovates on the online phase by applying \fourier series to \parseval's identity. \sysname enjoys of constant latency in the size of the circuit, but is only suitable to evaluate low-depth circuits. We introduce the first generalization of \parseval's identity for \fourier series applicable to an arbitrary number of inputs, and use it to allow \sysname to operate on an arbitrary number of inputs. \sysname paves the way for new kind of multiparty computation protocols, hopefully encouraging discussions and spurring new directions to explore.

\section*{Acknowledgements} This work is supported by the EU H2020 DECODE project under grant agreement number 732546 as well as \texttt{chainspace.io}. We thank George Danezis for helpful suggestions on early manuscript and valuable advice, and Yiannis Psaras for comments and proofreading.

\bibliographystyle{IEEEtranS}
\bibliography{references}

\appendix
\section{Finite Field Calculations} \label{sec:fields}
We show how the arithmetic described in \Cref{sec:construction} can be computed numerically using finite fields computations.
\Cref{eq:conv1} reads in complex variables as below:
\begin{align}\label{eq:FFC1}
&\phi(x)=\sum_{-\infty}^{\infty}A_n\exp{\bigl(2in\pi x/L\bigr)}\quad {\rm where}\\
&A_n=\frac{1}{L}\int_{-L/2}^{+L/2}\phi(x)\exp{\bigl(-2in\pi x/L\bigr)}dx\nonumber
\end{align}\label{eq:FFC1a}
Coefficients $A_n$ are linked to coefficients $a_n$ and $b_n$ by the following relations
\begin{equation}
A_n=
\left\{ \begin{array}{ll}
1/2(a_n+ib_n) \ & \ \mbox{{\rm if}\ \ $n < 0$}\\
a_0/2 \ & \ \mbox{{\rm if}\ \ $n=0$}\\
1/2(a_n-ib_n)\  & \ \mbox{{\rm if}\ \ $n > 0$}\\
\end{array}
\right.
\ \ {;\ \ \rm and}\ \ L=2l\nonumber
\end{equation}
For two \players, Parseval's identity reads
\begin{equation}\label{eq:FFC2}
\sum_{-\infty}^{\infty} A_n{\bar B}_n=\frac{1}{L}\int_{-L/2}^{+L/2}\phi(x){\bar\psi}(x)dx
\end{equation}
where $B_n$ is the complex Fourier transform of function $\psi(x)$ and the bar over the variables indicates the complex conjugate operation.
The discrete version of Equations \ref{eq:FFC1} and \ref{eq:FFC2} can be obtained by introducing a Galois Field. Let's assume a finite field ${\mathbb F}_{p^m}$ where $p$ is a prime and $m$ a positive integer, and let $N$ be a divisor of $p^m-1$ (possibly $N=p^m-1$); any periodic function $f(n)$ with period $N$ (\ie $f(n+N)=f(n)$ for all integers $n$) can be expressed as a finite Fourier series of the form
\begin{equation}\label{eq:FFC3}
f(n)=1/N \sum_{\kappa=0}^{N-1}g(\kappa) \exp{\bigl(2\pi i\kappa n/N\bigr)}
\end{equation}
with $\{g(k)\} \in \mathbb{F}_{p^m}$ with $k = 0, 1,\dots, N-1$. Note that $\exp{\bigl(2\pi i\kappa/N\bigr)}$ is a $N-th$ root of unity and the terms $\exp{\bigl(2\pi i\kappa n/N\bigr)}$ appearing in the expression for the finite \fourier series above are the $n-th$ powers of this. Note that \Cref{eq:FFC3} can also be evaluated outside the domain $\kappa\in [0, N-1]$ since the extended sequence is $N-$periodic. 
The fundamental theorem for discrete \fourier transforms also provides a simple formula for the coefficients $\{g(k)\}$~\cite{pollard,cooley}:
\begin{equation}\label{eq:FFC4}
g(\kappa)=\sum_{n=0}^{N-1}f(n)\exp{\bigl(-2\pi i\kappa n/N\bigr)}
\end{equation}
The validity of the periodicity property satisfied by $g(\kappa)$ (\ie $g(\kappa+N)=g(\kappa)$) and the following orthogonality relation
\begin{equation}\label{eq:FFC5}
1/N \sum_{n=0}^{N-1}\exp{\bigl(2\pi i(\kappa-\kappa')n/N\bigr)} =\delta_{\kappa\kappa'}
\end{equation}
are easily checked. \parseval's identity states
\begin{equation}\label{eq:FFC6}
\sum_{\kappa=0}^{N-1}A_\kappa {\bar B}_\kappa = N \sum_{n=0}^{N-1}\phi_n {\bar\psi}_n
\end{equation}
with $\{A_\kappa\}$ and $\{B_\kappa\}$ denoting the discrete \fourier transform coefficients of the sequences $\{\phi_n\}$ and $\{\psi_n\}$, respectively.
The convolution theorem for discrete \fourier transforms indicates that the convolution of two infinite sequences can be obtained as the inverse transform of the product of the individual transforms. A simplification occurs when the sequences are of finite length $N$. Let us suppose a sequence extended by periodic summation as
\begin{equation}\label{eq:FFC7}
(\phi_N)_n\equiv \sum_{\iota=-\infty}^{\infty}\phi_{(n-\iota N)}=\phi_{n (mod\ N)}	
\end{equation}
The convolution operation between two sequences $\{\phi_n\}$ and $\{\psi_n\}$, extended by periodic summation, is defined as 
\begin{equation}\label{eq:FFC8}
(\phi \star \psi_N)_n\equiv \sum_{\iota=0}^{N-1}\phi_\iota (\psi_N)_{n-\iota}
\end{equation}
We get:
\begin{equation}\label{eq:FFC9}
{\mathcal F}[(\phi \star \psi_N)_n]=(A\cdot B)_n
\end{equation}
with “$\mathcal F$” indicating the DFT operation and dot stands for {\it element-wise multiplication} between the two sequences $\{A_\kappa\}$ and $\{B_\kappa\}$. The convolution theorem duality reads
\begin{equation}\label{eq:FFC10}
{\mathcal F}[(\phi \cdot \psi)_\kappa]=\frac{1}{N}(A\star B_N)_\kappa
\end{equation}

Despite Discrete \fourier Transforms (DFT) are efficient to compute (\eg through Fast Fourier Transforms), it is important to note that a convenient choice of the mask-functions always allows to obtain the expressions of the \fourier transforms and \parseval's identity analytically (\ie without performing any numerical calculations); computing DFT is only needed if we operate using generic mask-functions. For instance, \Cref{sec:example} illustrates how to forgo numerical calculations using the evaluation of the convergent sums given in \Cref{sec:convergence}. \sysname ultimately only requires addition, multiplication and modular inversion over finite fields.

\end{document}